\documentclass[12pt, a4paper]{article}
\makeatother

\usepackage{amsmath}
\usepackage[english]{babel}
\usepackage{amsmath}
\usepackage{amsfonts}
\usepackage{amssymb}
\usepackage{graphicx}
\usepackage{amsthm,amssymb}
\usepackage{epsfig}
\usepackage{float}
\usepackage{color}
\usepackage[top=1in, bottom=1.2in, left=1.4in, right=1.4in]{geometry}
\usepackage{subfigure}
\numberwithin{figure}{section}
\newtheorem{theorem}{Theorem}[section]
\newtheorem{lemma}{Lemma}[section]

\newtheorem{proposition}{Proposition}[section]
\newtheorem{example}{Example}[section]

\newtheorem{definition}{Definition}[section]

\date{}
\usepackage{tikz}
\usetikzlibrary{matrix,arrows,shapes.geometric}
\usepackage{longtable}
\usepackage{soul}
\begin{document}
\title{Reversible cyclic codes over $\mathbb{F}_q + u \mathbb{F}_q$}
\author{\small{Om Prakash, Shikha Patel and Shikha Yadav} \\
		\small{Department of Mathematics} \\ \small{Indian Institute of Technology Patna} \\ \small{Bihta, Patna - 801 106, India} \\  \small{om@iitp.ac.in, shikha$\_$1821ma05@iitp.ac.in, 1821ma10@iitp.ac.in}}
\maketitle
\begin{center}
\textbf{Abstract}
\end{center}

	 Let $q$ be a power of a prime $p$. In this paper, we study reversible cyclic codes of arbitrary length over the ring $ R = \mathbb{F}_q + u \mathbb{F}_q$, where $u^2=0 ~mod~q$. First, we find a unique set of generators for cyclic codes over $R$, followed by a classification of reversible cyclic codes with respect to their generators. Also, under certain conditions, it is shown that dual of reversible cyclic code is reversible over $\mathbb{Z}_2+u\mathbb{Z}_2$. Further, to show the importance of these results, some examples of reversible cyclic codes are provided. \\

\hspace{-.7cm}\textbf{Keywords:} Linear codes; Cyclic codes; Hamming distance; Generator polynomial; Dual codes.\\

\hspace{-.7cm}\textbf{AMS Subject Classification:} 94B05, 94B15.
\section{Introduction}

 In algebraic coding theory, linear code over finite rings acquires intensive study in the last decades of $20$th century. This study over finite rings was prompted after the accomplishment of Gray maps. Remarkable steps came in $1994$ when Hammons et al. $\cite{16}$ obtained some good non-linear binary codes as an image of linear codes over $\mathbb{Z}_4$ under the Gray map. Afterward, the study of linear codes over finite rings have got more attention than a binary field and several families of codes were studied in \cite{22,T,12,20,23}, such as over $\mathbb{Z}_4,~ \mathbb{Z}_2+v\mathbb{Z}_2,v^2 =v;~ \mathbb{Z}_2+u\mathbb{Z}_2+v\mathbb{Z}_2+uv\mathbb{Z}_2, u^2 =v^2 =0;~ \mathbb{Z}_{p^r} +u\mathbb{Z}_{p^r} +\cdots+u^{k-1}\mathbb{Z}_{p^r}, u^k = 0,$ where $p$ is a prime. Note that cyclic codes are block linear codes in which the cyclic shift of each codeword is again a codeword. These error-correcting codes are also considered an important family of linear codes due to their rich algebraic structure, which makes this class easy to understand and implement. These codes have been studied over various finite rings and many new codes and results have been obtained in \cite{12,20,15,18}.

In $1970$, Hartmann and Tzeng \cite{18} have given bound for the minimum distance of certain reversible cyclic codes. In $2007$, Siap and Abualrub $\cite{14}$  studied the structure of reversible cyclic codes over $\mathbb{Z}_4$. In $2015$, Srinivasulu and Bhaintwal $\cite{11}$ studied reversible cyclic codes over $\mathbb{F}_4+u\mathbb{F}_4, u^2=0$ and their applications to DNA codes, meanwhile Sehmi et al. $\cite{19}$ studied reversible and reversible complement cyclic codes over Galois rings.

Motivated by these works, we study reversible cyclic codes of arbitrary length $n$ over $\mathbb{F}_q + u \mathbb{F}_q$, $u^2=0 ~mod~q$. Recall that these codes have applications in DNA computing which is a field of study that aims at harnessing individual molecules at the nanoscopic level for computational purposes. Computation with DNA molecules possesses an inherent interest for researchers in computer and biology. At present, many researchers have been interested in designing a new set of codewords for each experiment depending on various design constraints in DNA computing. One can prevent errors by minimizing the similarity between the sequences under some distance measure. These codes have many applications in constructing data storage and retrieval systems. \\

The presentation of the manuscript is as follows: In Section $2$, we give some preliminaries while Section $3$ provides the structure of cyclic codes of arbitrary length $n$ over the ring $R$. Section $4$ contains some important results on reversible cyclic codes over $\mathbb{F}_q + u \mathbb{F}_q$. In Section $5,$ some conditions are given under which dual of reversible cyclic code over $R$ is reversible. Section $6$ includes some examples in support of our results and Section $7$ concludes the work.

\section{Basic definitions and construction of cyclic codes over $ \mathbb{F}_q + u \mathbb{F}_q$ }
Throughout the article, $R= \mathbb{F}_q + u \mathbb{F}_q$, where $u^2=0~mod~q$ and $q= p^{k},$ a positive integer power of a prime $p$. Then $R$ is a commutative ring having $q^2$ elements.\\
Recall that a linear code $C$ of length $n$ over $R$ is an $R$-submodule of $R^n$ and cyclic code is a linear code invariant under the shift operator which maps $(c_0,c_1,\dots,c_{n-1})$ to $(c_{n-1},c_0,\dots,c_{n-2})$. Also, cyclic code over $R$ can be viewed as an ideal of $R_n=R[x]/ \langle x^n -1 \rangle $, identifying $(c_0,c_1,\dots,c_{n-1})$ by $c_0 +c_1x+\cdots+c_{n-1}x^{n-1}$. For $v=(v_0,v_1,\dots,v_{n-1}) \in R^n$, the vector obtained after reversal of components of $v$ is denoted by $v^r=(v_{n-1},v_{n-2}\dots,v_0)$.\\
The Hamming weight of a codeword is the number of non-zero components in it and the Hamming distance between any two codewords is number of components in which these two differ. The inner product of two vectors $a=(a_0,\dots,a_{n-1})$ and $b=(b_0,\dots,b_{n-1})$ is defined as $a \cdot b={\Sigma }_{i=0}^{n-1}a_i b_i$. The vectors $a$ and $b$ are said to be orthogonal if $a \cdot b=0$.
The dual $C^{\perp}$ of a linear code $C$ is defined as $C^{\perp} =\{ v \in R^n: v \cdot c=0 ~\text{for all } c \in C\}$. A linear code $C$ is said to be self dual if and only if $C=C^{\perp}$, and self-orthogonal if and only if $C\subseteq C^{\perp}$.
For each polynomial $f(x)=f_0 + f_1x+ \cdots + f_{n-1}x^{n-1}$ with $f_{n-1}\neq 0$, the reciprocal of $f(x)$ is defined as $f^*(x)=x^{n-1}f(1/x)=f_{n-1} + f_{n-2}x+\cdots+ f_{0}x^{n-1}$. Note that $deg f^*(x) \leq deg f(x)$, and if $f_0 \neq 0$, then $deg f^*(x)= deg f(x)$. The polynomial $f(x)$ is called self-reciprocal if and only if $f^*(x)= f(x)$.

Let $C$ be a cyclic code over $R$. Then a map $\phi:C \rightarrow \mathbb{F}_q[x]/\langle x^n-1 \rangle$ defined by $\phi(a_0+a_1x+\cdots+ a_{n-1}x^{n-1})=a_0^q+a_1^qx+\cdots+ a_{n-1}^qx^{n-1}$ is a ring homomorphism with $ker\phi=\{ur(x)| r(x) ~\text{is a polynomial in} ~  \mathbb{F}_q[x]/\langle x^n-1 \rangle\}$. Let $J=\{r(x) ~ : ~ur(x) \in ker\phi \}$. Then $J$ is a cyclic code over $\mathbb{F}_q$, being an ideal of $\mathbb{F}_q[x]/\langle x^n-1 \rangle$. Therefore, $J=\langle a(x) \rangle ~ \text{where}~ a(x)|(x^n-1)$. This implies $ ker\phi = \langle ua(x) \rangle ~\text{where}~ a(x)|(x^n-1)~ mod~q$. Since image of $\phi$ is an ideal and hence a cyclic code over $\mathbb{F}_q$ with generator polynomial $g(x)$ such that $g(x)|(x^n-1)$. Hence,
 $$C = \langle g(x)+up(x), ua(x) \rangle$$ for some polynomial $p(x)$ over $\mathbb{F}_q$.

 Throughout the article, we use same $g(x), ~p(x)$ and $a(x)$ as mentioned above. Now, we will give some lemmas and theorem having proof with similar arguments as given in \cite{12}, and will be used later for the discussion on reversible cyclic code.

\begin{lemma}
For above $a(x)$ and $p(x),$ $deg~ a(x)>deg ~p(x)$ and $a(x)\mid g(x)$.
\end{lemma}
\begin{proof}
Note that
 $$ C = \langle g(x)+up(x), ua(x) \rangle = \langle g(x)+u(p(x)+x^i a(x)), ua(x) \rangle.$$ Meanwhile, $$ \langle g(x)+up(x), ua(x) \rangle = \langle g(x)+u(p(x)+d(x) a(x)), ua(x) \rangle.$$ Therefore, from above, we may assume $deg(p(x)) < deg (a(x))$. Also, $$ug(x) \in ker \phi = \langle ua(x) \rangle$$ implies that $a(x)|g(x)$. If $g(x)= a(x)$, then $C=\langle g(x)+up(x) \rangle $.
\end{proof}
\begin{lemma}
$a(x)$ divides $p(x) \left(\frac{ x^n-1}{g(x)}\right)$.
\end{lemma}
\begin{proof}
Since $$\phi\left(\frac{x^n-1}{g(x)}(g(x)+up(x))\right)=\phi\left( up(x) \frac{x^n-1}{g(x)}\right)=0.$$
This implies $\left(up(x)\frac{x^n-1}{g(x)}\right)\in ker\phi = \langle ua(x) \rangle$. Therefore, $a(x)|\left(p(x)\frac{x^n-1}{g(x)}\right).$
\end{proof}
\begin{lemma}
Let $C = \langle g(x)+up(x), ua(x) \rangle =  \langle h(x)+uq(x), ub(x) \rangle $. Then $g(x)=h(x), a(x)=b(x)$ and $p(x)=q(x)~ mod ~a(x)$.
\end{lemma}
\begin{proof}
From the construction of $C$, we have $J=\langle a(x) \rangle=\langle b(x) \rangle$, i.e., $a(x)=b(x)$.
Let $C = \langle g(x)+up(x), ua(x) \rangle= \langle h(x)+uq(x), ub(x) \rangle$.
Since $h(x)\in \phi(C)= \langle g(x) \rangle $, we have  $$h(x)=g(x)\alpha (x)~ \text{ and }deg~h(x) \geq deg~g(x).$$ Similarly, $$g(x)=h(x)\beta (x)=g(x)\alpha(x)\beta(x) ~\text{and}~ deg~g(x) \geq deg~h(x).$$ Now, $(x^n-1)$ factors uniquely into irreducible polynomials over $\mathbb{F}_q$ and $g(x), h(x)$ are monic polynomials which divides $(x^n-1)$, therefore $\alpha(x)=\beta(x)=1$ and $g(x)=h(x)$. Since $g(x)+uq(x) \in C$, we have $$g(x)+uq(x)=(g(x)+up(x))+ua(x)l(x), ~\text{for some}~l(x)\in R[x].$$ This implies $u(q(x)-p(x))=ua(x)l(x)$ and hence $p(x)=q(x) ~mod ~a(x).$
\end{proof}
\begin{lemma}
If $n$ is relatively prime to $q$, then $C=\langle g(x),ua(x)\rangle =\langle g(x)+ua(x)\rangle$.
\end{lemma}
\begin{proof}
Let $a(x)|g(x)$ and $a(x)|p(x) \left(\frac{ x^n-1}{g(x)}\right)$. Then $g(x)=a(x)l_1(x)$ and\\ $p(x) \left(\frac{ x^n-1}{g(x)}\right)=a(x)l_2(x)$. Since $n$ is relatively prime to $q$, $x^n-1$ can uniquely be written as product of distinct irreducible polynomials and hence $a(x)$ must be a factor of $p(x)$. But $deg~p(x) < deg~a(x),$ therefore, $p(x)=0$ and $C=\langle g(x),ua(x)\rangle$.

Let $b(x)=g(x)+ua(x)$. Then $ub(x)=ug(x) \in \langle g(x)+ua(x)\rangle$ and $\left(\frac{x^n-1}{g(x)}\right)b(x)=u\left(\frac{x^n-1}{g(x)}\right)a(x) \in \langle g(x)+ua(x)\rangle$. Since $gcd~\left(\frac{x^n-1}{g(x)}, g(x)\right)=1$, there exist polynomials $g_1(x)$ and $g_2(x)$ over $\mathbb{F}_q$ such that
\begin{align*}
  1&=\frac{x^n-1}{g(x)}g_1(x)+ g(x)g_2(x)
\\& ua(x)=u\left(\frac{x^n-1}{g(x)}\right)a(x)g_1(x)+ ug(x)a(x)g_2(x)
\in \langle g(x)+ua(x)\rangle.
\end{align*}
Also, $g(x)=b(x)-ua(x) \in\langle g(x)+ua(x)\rangle$. Therefore, $C=\langle g(x), ua(x)\rangle=\langle g(x)+ua(x)\rangle.$
\end{proof}

\begin{theorem}\label{Th1}
Let $C$ be a cyclic code of length $n$ over $R$.
\begin{enumerate}

\item If $n$ is relatively prime to $q$, then $R[x]/ \langle x^n - 1 \rangle$ is a principal ideal ring and $C = \langle g(x),ua(x) \rangle = \langle g(x)+ua(x) \rangle$
where $g(x)$, $a(x)$ are polynomials over $\mathbb{F}_q$ with $ a(x)|g(x)|(x^n -1)~ mod~ q$.
\item  If $n$ is not relatively prime to $q$, then

\item[(a)] $C = \langle g(x)+up(x) \rangle$ where $g(x)$, $p(x)$ are polynomials over $\mathbb{F}_q$ with $g(x)|(x^n -1) ~ mod~ q $, $ (g(x)+up(x))|(x^n -1)$ and $g(x)|p(x) \left( \frac{ x^n-1}{g(x)}\right)$ and also $g(x)=a(x)$.
\item[(b)] $C= \langle g(x) + up(x), ua(x) \rangle $ where $g(x)$, $a(x)$, and $p(x)$ are polynomials over $\mathbb{F}_q$ with $ a(x)|g(x)|(x^n -1)~ mod~ q$, $ (g(x)+up(x))|(x^n -1)$, $a(x)|p(x) \left(\frac{ x^n-1}{g(x)}\right)$ and
$ deg(g(x)) > deg(a(x)) > deg(p(x))$.

\end{enumerate}
\end{theorem}

\section{Reversible cyclic code over $ R$}
In this section, we study reversible codes separately for even and odd lengths and find necessary and sufficient condition
for a cyclic code $C$ over $R$ to be reversible.  The reverse of a codeword $c = (c_0, c_1,\dots, c_{n-1}) \in  C$ is
denoted by $c^r$, and defined as $c^r =(c_{n-1}, c_{n-2},\dots, c_0)$.

\begin{definition} A linear code $C$ of length $n$ over a ring $R$ is said to
be reversible if $c^r \in C$, for all $c \in C$.
\end{definition}

The following theorem characterizes a cyclic code to be reversible over the finite field.
\begin{theorem}\label{4} $\cite[Theorem~1]{15}$
The cyclic code over $GF(q)$ generated by the monic polynomial $g(x)$ is reversible if and only if $g(x)$ is self-reciprocal.
\end{theorem}

\begin{lemma}\label{1} $\cite[Lemma~19]{22}$  Let $f (x), g(x)$ be any two polynomials
in $R[x]$ with $deg( f ) \geq deg(g)$. Then

\begin{enumerate}
    \item $( f (x)g(x))^{*} = f^{*}(x)g^{*}(x)$;
\item $( f (x)+g(x))^{*} = f ^{*}(x)+x^{deg( f )-deg(g)}g^{*}(x)$.
\end{enumerate}
\end{lemma}
\begin{lemma}\label{2}
  Let $C$ be a reversible cyclic code of length $n$ over the ring $R$ and $\phi: C \rightarrow  \frac{F_q[x]}{\langle x^n -1\rangle}$ as defined in the Section $2$, be a ring homomorphism. Then $\phi(C) $ is reversible.
 \end{lemma}
 \begin{proof}
 Let $\phi(c) \in \phi(C)$, where $c=(c_0,c_1,\dots,c_{n-1}) \in C$, i.e., $\phi(c)=(c_0^q,c_1^q,\dots\\,c_{n-1}^q)\in \phi(C)$.
Since $C$ is a reversible cyclic code, $c^r=(c_{n-1},c_{n-2},\dots,c_0) \in C$. Consider
\begin{align*}
 \phi(c)^r&=   (c_0^q,c_1^q,\dots,c_{n-1}^q)^r \\
 &=(c_{n-1}^q,c_{n-2}^q, \dots,c_0^q)\\ &=\phi(c_{n-1},c_{n-2},\dots,c_0) \in \phi(C).
 \end{align*}
Hence, $\phi(C)$ is reversible.

 \end{proof}

\begin{lemma}\label{3}
Let $C$ be a reversible cyclic code over $R$. Then $\langle g(x) \rangle $ and $\langle a(x) \rangle$ are also reversible cyclic codes over $\mathbb{F}_q.$
\end{lemma}
\begin{proof}
From the construction of generators of cyclic codes over $R$, we have $\phi(C)=\langle g(x) \rangle$ and by Lemma \ref{2}, $\phi(C)$ is a reversible code over $\mathbb{F}_q$. Therefore, $\langle g(x) \rangle$ is reversible cyclic code over $\mathbb{F}_q$.

As $ker(\phi)=\langle ur(x) | r(x)$ is a polynomial in $C$ with coefficients in $\mathbb{F}_q \rangle$ and $J=\langle r(x) |ur(x)
\in ker(\phi) \rangle = \langle a(x) \rangle$, it is sufficient to show that $J$ is reversible. Let $r(x)=r_0 + r_1x+...+r_{n-1}x^{n-1} \in J$ be arbitrary,
 then $r(x) \in \mathbb{F}_q[x]$ is a polynomial in $C$ . Since $C$ is reversible cyclic code in $R$, therefore $r^*(x)$ is also in $C$. Also, $ur^*(x) \in ker(\phi)$ i.e., $r^*(x) \in J$. Hence, we get the required result.

\end{proof}

\begin{theorem}\label{6} Let $C = \langle g(x), ua(x)\rangle$  be a linear cyclic code of odd length $n$ over $R$, where $ a(x) |g(x) | (x^n -1)$ and
$a(x), g(x) \in \mathbb{F}_q[x].$ Then $C$ is reversible if and only if both
$g(x)$ and $a(x)$ are self reciprocal.
\end{theorem}
\begin{proof}
Let $C$ be a reversible cyclic code over $R$. Then by Lemma $\ref{3}$ and Theorem $\ref{4}$, $g(x)$ and $a(x)$ are self-reciprocal polynomials.

For sufficient part, we assume that $g(x)$ and $a(x)$ are self-reciprocal polynomials over
$\mathbb{F}_q$.
Let $c(x) \in C$, i.e., $c(x) = g(x)m_1(x)+ua(x)m_2(x)$ for some
polynomials $m_1(x)$ and $m_2(x)$ over $R$. $C$ is reversible if and only if $c^*(x) \in C$. Consider

\begin{align*}
c^{*}(x) &=(g(x)m_1(x)+ua(x)m_2(x))^{*}\\
&= (g^{*}(x)m_1^{*}(x)+ux^{i}a^{*}(x)m_2^{*}(x))\\
&= (g(x)m_1^{*}(x)+ua(x)x^{i}m_2^{*}(x)),
 \end{align*}
where $m_1^*(x), m_2^*(x)$ are polynomials over $R$. This implies  $c^{*}(x) \in \langle g(x),\\ u a(x) \rangle $. Thus, $C$ is a reversible cyclic code over $R$.
\end{proof}

 \begin{theorem} \label{7}
Let $C = \langle g(x) + up(x), ua(x)\rangle $ be a cyclic
code of even length $n$ over $R$  where $a(x), g(x)$ and $p(x)$
are polynomials over $ \mathbb{F}_{q} $ such that $deg~a(x) > deg~p(x)$,
$a(x)|g(x)|(x^{n}-1)$  and $a(x)|p(x)( \frac{ x^n-1}{g(x)})$. Then $C$ is reversible if and only if
\begin{enumerate}
    \item $g(x)$ and $a(x)$ are self-reciprocal, and
\item $a(x)$ divides $(x^ip^{*}(x)-p(x))$, where $i = deg~g(x)-deg~p(x)$.
\end{enumerate}

\end{theorem}
\begin{proof}
Let $C =\langle g(x)+up(x), ua(x)\rangle$ be reversible cyclic code over $R$. Then $ g(x)$ and $a(x)$ are self-reciprocal by Lemma \ref{3} and Theorem \ref{4}. Now, consider
\begin{equation*} \label{8}
[g(x)+up(x)]^* =[g^*(x)+ux^ip^*(x)] = g(x)+ux^ip^*(x),
\end{equation*} where $i=deg~g(x)-deg~p(x)$. Since $C$ is reversible cyclic code over $R$, we have $g(x)+ux^ip^*(x)\in C$. Therefore, there exist $l_1(x)$ and $l_2(x)$ in $R[x]$ such that
\begin{equation*} \label{9}
g(x)+ux^ip^*(x)=[g(x)+up(x)]l_1(x)+ua(x)l_2(x).
\end{equation*}
Comparing the degrees on both sides, we get $l_1(x)$ is a constant over $R,$ say, $l_1(x)=a+ub$ where $a, b \in \mathbb{F}_q$. Then
\begin{equation*} g(x)+ux^ip^*(x)=ag(x)+bug(x)+aup(x)+ua(x)l_2(x).
\end{equation*}
Multiplying the above equation by $u$, we get $ug(x)=uag(x)$. Therefore, $l_1(x)=1+ub \in R$. Since $a(x)|g(x)$, it can be easily seen that $a(x)|(x^ip^*(x)-p(x))$.

Conversely, assume that conditions (1) and (2) hold. Let $c(x) \in C$. Then $c(x)=(g(x)+up(x))l_1(x) + ua(x)l_2(x)$, for some polynomials $l_1(x)$ and  $l_2(x)$ over $R$. Consider
\begin{align*}
c^*(x)&=(g(x)+up(x))^* l_1^*(x) + ux^ja^*(x)l_2^*(x)
\end{align*}
\begin{align}\label{eq8}
c^*(x)&=(g(x)+ux^ip^*(x))l_1^*(x) + ux^ja(x)l_2^*(x).
\end{align}
Since $a(x)$ divides $p(x)+x^ip^{*}(x)$, so $p(x)+x^ip^{*}(x)=a(x)b(x)$ for some polynomial $b(x)$ over $R$. Also, $ux^ip^{*}(x)=up(x)+ua(x)b(x)$. Substituting in equation (\ref{eq8}), we have

\begin{align*}
c^*(x)=(g(x)+up(x)+ua(x)b(x))l_1^*(x) + ux^ja(x)l_2^*(x) \\ = (g(x)+up(x))l_1^*(x)+ua(x)(b(x)l_1^*(x)+x^jl_2^*(x)).
\end{align*}
Therefore,  $c^*(x) \in C$. Hence, $C$ is reversible cyclic code over $R$.

\end{proof}
\begin{theorem}\label{10}
Let $C = \langle g(x)+up(x) \rangle $ be a cyclic code of
even length $n$ over $R$. Then $C$ is reversible if and only if
\begin{enumerate}
    \item $g(x)$ is self-reciprocal;
\item $p(x) = x^i p^{*}(x)$ or $g(x) = b^{-1}[x^ip^{*}(x)-p(x)]$, where
$i = deg~g(x)-deg~p(x)$.
\end{enumerate}
\end{theorem}
\begin{proof}
Let $C$ be reversible cyclic code over $R$. Then $\langle g(x) \rangle $ is reversible cyclic code over $\mathbb{F}_q$. By Theorem \ref{4} , $g(x)$ is self reciprocal.\\
For part (2), consider
\begin{align*} [g(x)+up(x)]^{*}= [g^{*}(x)+ux^{i}p^{*}(x)]\\
= g(x)+u x^ip^{*}(x).
 \end{align*}
Since $C$ is reversible, we have $g(x)+ux^ip^{*}(x) \in C$. This implies there exists a polynomial $l(x)$ over $\mathbb{F}_q$ such that
\begin{equation*} g(x)+ux^i p^{*}(x) = [g(x)+up(x)]l(x).
\end{equation*}
Comparing the degrees on both sides, we get $l(x)$ is a constant polynomial over $R,$ say, $l(x)=a+ub$, where $a, b \in \mathbb{F}_q$. Then
\begin{equation*} g(x)+ux^ip^*(x)=ag(x)+bug(x)+aup(x). \end{equation*}
Multiplying the above equation by $u$, we have $ug(x)=uag(x)$
and hence $a=1$. If $b=0,$ then $p(x)=x^ip^*(x)$ and if $b\neq 0,$ then $g(x)=p(x)+x^ip^*(x)$.

Conversely, assume that $(1)$  and $(2)$  hold. Let $c(x) \in C$, i.e., $c(x)=(g(x)+up(x))r(x)$, for some polynomials $r(x)$ over $R$. Consider
\begin{equation} \label{eq1}
c^*(x)=[g(x)^*+ux^ip^*(x)]r^*(x)=[g(x)+ux^ip^*(x)]r^*(x).
\end{equation}
\textbf{Case 1} If $p(x)=x^ip^{*}(x)$, then
$$c^*(x)=[g(x)+up(x)]r^*(x) \in C$$
 \textbf{Case 2} If $g(x)=b^{-1}[x^ip^{*}(x)-p(x)]$, being $g(x)$ self-reciprocal, we have
\begin{equation*}
  x^ip^{*}(x)-p(x)=x^{i+j}p(x)-p^*(x)
\end{equation*} i.e.,
\begin{equation}\label{23}
    up(x)(1+x^{i+j})-up^*(x)=ux^ip^*(x).
\end{equation}
Using (\ref{23}) in (\ref{eq1}), we have
\begin{align*}
    & [g(x)+up(x)+u(x^{i+j}p(x)-p^*(x))]r^*(x)
\\&=[(1+u)g(x)+up(x)]r^*(x)\\&=(1+u)[g(x)+(1+u)up(x)]r^*(x)\\&= [g(x)+up(x)]s(x),
\end{align*} where $s(x)=(1+u)r^*(x)$. Therefore, $c^*(x)\in C$ and hence $C$ is reversible cyclic code over $R$.
\end{proof}
\section{Dual of a reversible cyclic code over $\mathbb{Z}_2
 +u \mathbb{Z}_2 $}
Let $C$ be a cyclic $[n,k]$-code with parity check polynomial $h(x)=h_0+h_1x+\cdots+ h_kx^k$ and $\bar{h}(x)=h^*(x)$. Then we have the following characterization for $C^\perp$ dual code of $C$.
\begin{theorem}\label{10}
Let $C^\perp$ be a dual code of a cyclic code  $C$ over $GF(q).$ Then $C^\perp=\langle \bar{h}(x)\rangle$ is reversible cyclic code if and only if $h(x)\in C^\perp$.
\end{theorem}
\begin{proof}
Let $\bar{h}(x)=(h_{k},h_{k-1},\dots, h_0)\in C^\perp$. Then $(\bar{h}(x))^r=(h_0,h_{1},\dots, h_{k-1},\\h_{k})=h(x)\in C^\perp$.

Conversely, suppose $h(x)\in C^\perp$, then $h(x)=(\bar{h}(x))^r$, i.e., $(\bar{h}(x))^r=h(x)\in C^\perp$. Therefore, $C^\perp=\langle \bar{h}(x)\rangle$ is reversible cyclic code.

\end{proof}
\begin{definition}
Let $I$ be an ideal in $R_n$. The annihilator $A(I)$ of $I$ in $R_n$ is defined as
$$A(I)= \{b(x)~|~ f(x)b(x)=0 ~\text{for all} ~f(x) ~\text{in}~ I\}. $$
\end{definition}
If $C$ is a cyclic code with associated ideal $I,$ then the  associated ideal of $C^\perp$ is
$$A(I)^*= \{g^*(x)~|~ g(x)\in I\}.$$
\begin{proposition}\label{prop1}
Let $C$ be a cyclic code of odd length over $\mathbb{Z}_2+u\mathbb{Z}_2$. Then
$$Ann(C)=\left \langle \frac{x^n-1}{a(x)},u\frac{x^n-1}{g(x)} \right \rangle.$$
\end{proposition}
\begin{proof}
Since $C$ is a cyclic code of odd length over $\mathbb{Z}_2+u\mathbb{Z}_2$, we have $$C= \langle g(x)+ua(x) \rangle=\langle g(x), ua(x) \rangle $$ with $a(x)|g(x)|(x^n-1)$. Also, there exists $m_1(x)$ such that $g(x)=a(x)m_1(x)$.
Note that
\begin{align*}
    \left(\frac{x^n-1}{a(x)}\right)(g(x)+ua(x))&=\left(\frac{x^n-1}{a(x)}\right)g(x)+u\left(\frac{x^n-1}{a(x)}\right)a(x)\\
    &=\left(\frac{x^n-1}{a(x)}\right)a(x)m_1(x)=0.
\end{align*}
Also, $$u\frac{x^n-1}{g(x)}(g(x)+ua(x)) =0.$$
Hence, $M=\left \langle \frac{x^n-1}{a(x)},u\frac{x^n-1}{g(x)}\right\rangle \subseteq Ann(C)$. \\
In order to prove $Ann(C)\subseteq M$, let $Ann(C)=\langle h(x),ur(x)\rangle$. Then
$$ur(x)(g(x)+ua(x))=0.$$ This implies there exists a polynomial $t_1(x)$ in $\mathbb{Z}_2$ such that
$$r(x)=\left(\frac{x^n-1}{g(x)}\right)t_1(x)\in M.$$
Also,
\begin{align*}
&h(x)(g(x)+ua(x))=0\\
&h(x)g(x)+uh(x)a(x)=0.
\end{align*}
Since $h(x)g(x)=0$. So,
$uh(x)a(x)=0$, i.e., there exists polynomial $t_2(x)$ in $\mathbb{Z}_2$ such that
$$h(x)=\left(\frac{x^n-1}{a(x)}\right)t_2(x).$$
Hence,
$$Ann(C)=\langle h(x),ur(x)\rangle
\subseteq\left \langle \frac{x^n-1}{a(x)},u\frac{x^n-1}{g(x)} \right \rangle\in M.$$ Therefore, $$Ann(C)=\left \langle \frac{x^n-1}{a(x)},u\frac{x^n-1}{g(x)} \right \rangle.$$
\end{proof}
The following result is the consequence of Proposition $\ref{prop1}$.
\begin{theorem}
  Let $C$ be a cyclic code of odd length over $\mathbb{Z}_2+u\mathbb{Z}_2$. Then
  $$C^{\perp}=\left \langle \left(\frac{x^n-1}{a(x)}\right)^* ,u\left(\frac{x^n-1}{g(x)}\right)^*\right \rangle.$$
\end{theorem}
\begin{theorem}
  Let $C$ be a reversible cyclic code of odd length $n$ over $\mathbb{Z}_2+u\mathbb{Z}_2$ with $a(x)|g(x)|(x^n-1)$ and  $C^{\perp}= \left \langle \left(\frac{x^n-1}{a(x)}\right)^*,u\left(\frac{x^n-1}{g(x)}\right)^* \right \rangle$. Then $C^{\perp}$ is a reversible cyclic code over $\mathbb{Z}_2+u\mathbb{Z}_2$.
\end{theorem}
\begin{proof}
Let $C$ be a reversible cyclic code of odd length $n$ over $\mathbb{Z}_2+u\mathbb{Z}_2$. Then $g(x)$ and $a(x)$ are self reciprocal. Assume $\left(\frac{x^n-1}{a(x)}\right)=r_1(x)$ and $\left(\frac{x^n-1}{g(x)}\right)=r_2(x)$. Now,\\
$$(x^n-1)^{*}=a^{*}(x)r^{*}_1(x) $$ and $$(x^n-1)^{*}=g^{*}(x)r^{*}_2(x). $$  This implies $r^{*}_1(x)=\frac{(x^n-1)^{*}}{a^{*}(x)}=\frac{-(x^n-1)}{a(x)}=-r_1(x)$ and $r^{*}_2(x)=\frac{(x^n-1)^{*}}{g^{*}(x)}=\frac{-(x^n-1)}{g(x)}=-r_2(x)$.\\
 Let $\bar{c}(x)\in C^{\perp}$. Then
\begin{align*}
(\bar{c}(x))^{*}&=\left(\left(\frac{x^n-1}{a(x)}\right)^{*}l_1(x)+\left(\frac{x^n-1}{a(x)}\right)^{*}l_2(x)\right)^{*}\\
&=\left(-r_1(x)l_1(x)-r_2(x)l_2(x)\right)^{*}\\
&=-r^{*}_1(x)l^{*}_1(x)-x^{i}r^{*}_2(x)l^{*}_2(x)\\
&=r^{*}_1(x)q_1(x)+r^{*}_2(x)q_2(x),
\end{align*}
for some polynomials $q_1(x)=-l^{*}_1(x)$ and  $q_2(x)=-x^{i}l^{*}_2(x)$ over $R$. Therefore, $c(x)\in C^{\perp}$. Thus, by Theorem \ref{10}, $C^{\perp}$ is a reversible cyclic code over $\mathbb{Z}_2+u\mathbb{Z}_2$.
\end{proof}

Now, we present a result given by Abualrub and Siap \cite{T}, which is used for furtherance on the dual of a reversible cyclic code.
\begin{theorem}$\cite[Theorem~4]{T}$ Let $C$ be a cyclic code of even length over $\mathbb{Z}_2+u\mathbb{Z}_2$.
\begin{enumerate}

    \item If $C=\langle g(x)+up(x), ua(x) \rangle $, with $a(x)|g(x)|(x^n-1)$, $a(x)|p(x)\left(\frac{x^n-1}{g(x)}\right)$ and $deg~ g(x)>deg ~a(x)> deg~p(x)$, and $g(x)=a(x)m_1(x)$, $p(x)\left(\frac{x^n-1}{g(x)}\right)\\=a(x)m_2(x)$, then $$Ann(C)=\left \langle \frac{x^n-1}{a(x)}+um_2(x),u\frac{x^n-1}{g(x)} \right \rangle ~\text{and}$$ $$C^{\perp}=\left \langle \left(\frac{x^n-1}{a(x)}\right)^*+ux^im_2^{*}(x),u\left(\frac{x^n-1}{g(x)}\right)^* \right \rangle,$$ where $i=deg \left(\frac{x^n-1}{a(x)}\right) - deg(m_2(x))$.
    \item If $C=\langle g(x)+up(x) \rangle$ with $p(x)\left(\frac{x^n-1}{a(x)}\right)=g(x)m_2(x)$, then $$Ann(C)=\left \langle \frac{x^n-1}{g(x)}+um_2(x)\right \rangle ~\text{and}$$
    $$C^{\perp}=\left \langle \left(\frac{x^n-1}{g(x)}\right)^*+ux^im_2^{*}(x) \right \rangle,$$ where  $i=deg \left(\frac{x^n-1}{g(x)}\right) - deg(m_2(x))$.
\end{enumerate}
\end{theorem}
\begin{theorem}
Let $C$ be a reversible cyclic code of even length $n$ over $\mathbb{Z}_2+u\mathbb{Z}_2$ and $C^{\perp}= \left \langle \left(\frac{x^n-1}{a(x)}\right)^*+ux^i(m_2^*(x)),~ u\left(\frac{x^n-1}{a(x)}\right)^*\right \rangle$. If $a(x)$ divides $p^*(x)+x^jp(x)$, where $i= deg \left(\frac{x^n-1}{a(x)}\right) - deg(m_2(x))$ and $j=deg \left(\frac{x^n-1}{a(x)}\right)^* - deg(m_2^*(x))-i,$ then $C^{\perp}$ is reversible cyclic code over $\mathbb{Z}_2+u\mathbb{Z}_2$ provided $p(x)\neq 0$.
\end{theorem}
\begin{proof}
Let $C$ be a reversible cyclic code of even length $n$ over $\mathbb{Z}_2+u\mathbb{Z}_2$. For $C^{\perp}$ to be reversible, it suffices to show that $\left(\frac{x^n-1}{a(x)}\right)+ux^{i+j}m_2(x)$ and $u\left(\frac{x^n-1}{a(x)}\right)$ are in $C^{\perp}$. Note that
\begin{align*}
 &\left(\left(\frac{x^n-1}{a(x)}\right)+ux^{i+j}m_2(x)\right)(g(x)+up(x))\\&= up(x)\left(\frac{x^n-1}{a(x)}\right) + ux^{i+j}m_2(x)g(x)\\&=um_2(x)g(x)+ux^{i+j}m_2(x)g(x)\\&= um_2(x)g(x)(1+x^{i+j}).
 \end{align*}
Since $p(x)\left(\frac{x^n-1}{g(x)}\right)=a(x)m_2(x)$ and $a(x)$ divides $g(x)$, we get
\begin{align*}
    & um_2(x)g(x)(1+x^{i+j})=\left(\frac{x^n-1}{a(x)}\right)up(x)(1+x^{i+j})
  \\&=(x^n-1)u\left(\frac{p(x)+x^{i+j}p(x)}{a(x)}\right)
  \\&=(x^n-1)u\left(\frac{p(x)+x^ip^*(x)+x^i(p^*(x)+x^jp(x))}{a(x)}\right).
  \end{align*}
By Theorem \ref{7}, $a(x)$ divides $p(x)+x^ip^*(x)$. Hence, above expression can be written as
\begin{equation*}
 (x^n-1)u\left(\frac{a(x)l^{'}(x)}{a(x)}\right)
=(x^n-1)ul^{'}(x)=0.
\end{equation*}
Next, we have $$\left(\left(\frac{x^n-1}{a(x)}\right)+ux^{i+j}m_2(x)\right)(g(x)+up(x))=0,$$
\\and $$\left(\left(\frac{x^n-1}{a(x)}\right)+ux^{i+j}m_2(x)\right)ua(x)=0,$$\\
and $$u\left(\frac{x^n-1}{a(x)}\right)(g(x)+up(x))=0,$$\\
and $$u\left(\frac{x^n-1}{a(x)}\right)ua(x)=0.$$\\
Hence, $C^{\perp}$ is a reversible cyclic code over $\mathbb{Z}_2+u\mathbb{Z}_2$.

\end{proof}
\section{Minimum Hamming distance of a cyclic code over $R $}
In this section, we find the minimum Hamming distance of a cyclic code of arbitrary length over $R$. Let $C=\langle g(x)+up(x), ua(x) \rangle$ be a cyclic code of length $n$ over $R$. Define $C_u=\{b(x)| ub(x) \in C\}$. Then $C_u$ is a cyclic code of length over $n$ over $\mathbb{F}_q$.	The following results give a technique to find minimum distance of a cyclic code of arbitrary length over $R$.
	\begin{theorem}
	Let $C=\langle g(x)+up(x), ua(x) \rangle$ be a cyclic code of length $n$ over $R$. Then $C_u=\langle a(x) \rangle$.
\end{theorem}
\begin{proof}
Straightforward.
\end{proof}
\begin{theorem}
	Let $C$ be a cyclic code of length $n$ over $R$. Then $d_{H}(C)=d_{H}(C_u)$.
	\end{theorem}
	\begin{proof}
Let $m(x)\in C_u$ be such that $d_{H}(C_u)=w_{H}(m(x))$. Then $um(x)\in C$. Also, $w_{H}(m(x))=w_{H}(um(x))$, hence, $d_{H}(C_u) \geq d_{H}(C)$. \\
Conversely, suppose $d_{H}(C)=w_{H}(b(x))$ where $b(x) \in C$. If the coefficient of any power of $x$ in $b(x)$ is a zero divisor of $R ,$ then we get an element in $C$ with hamming weight less than that of $b(x),$ which contradicts our assumption. Therefore, the coefficients of any power of $x$ in $b(x)$ is either zero or a unit in $R$, i.e., $\alpha +u \beta~ \text{such that either}~ \alpha= \beta=0 ~\text{or}~ \alpha \in \mathbb{F}_q^*$. In this case, $ub(x)=u c(x)$ for some $c(x) \in C_u$. Therefore, $w_{H}(b(x))=w_{H}(uc(x)) = w_{H}(c(x))$ and hence $d_{H}(C_u) \leq d_{H}(C)$. Thus, $d_{H}(C)=d_{H}(C_u)$.
	\end{proof}

\section{Examples}
\begin{example} For length $n=4$.
$$x^4-1=(x+1)(x+2)(x^2+1) ~\text{over}~ \mathbb{Z}_3.$$
 Some of the reversible cyclic codes of length $4$ over $\mathbb{Z}_3
 +u \mathbb{Z}_3 $ are given below:
\begin{longtable}[h]{|l|l|l|l|l|l|}
\hline
	  Non-zero Generator & Dimension $k$ &$d(C)$   & MDS \\
	 	Polynomial (s) of $C$ & of $C $ &  &\\
	 \hline
	 $1$ or $1+u$  & $4$ &  $1$ & $*$\\
\hline
$x+1$  & $3$ & $2$ &$*$\\
\hline
 $x^2+1$  & $2$ & $2$ &\\
\hline
 $(x+1)(x^2+1)$ & $1$ & $4$ & $*$\\
\hline
  $x+1, u$ & $4$ & $1$ & $*$\\
\hline
   $x^2+1, u$ & $4$ & $1$ & $*$\\
\hline
  $(x+1)(x^2+1), u$ & $4$ & $1$ & $*$\\
\hline
 $(x+1)(x^2+1), u(x+1)$ & $3$ & $2$ & $*$\\
\hline
  $(x+1)(x^2+1), u(x^2+1)$ & $2$ & $2$ &\\
\hline
\end{longtable}
\end{example}

\begin{example} For length $n=5$.
$$x^5-1=(x-1)(x^4+x^3+x^2+x+1) ~\text{over}~ \mathbb{Z}_2.$$

Since all of the above factors are self reciprocal polynomials, by Theorem $\ref{6},$ all the cyclic codes of length $5$ over $R$ are reversible. Some of these are given below:
\begin{longtable}[h]{|l|l|l|l|l|l|}
\hline
	  Non-zero Generator & Dimension $k$ &$d(C)$   & MDS \\
	 	Polynomial (s) of $C$ & of $C $ &  &\\
	
\hline
 $1$ or $(1+u)$  & $5$ &  $1$ & $*$\\
\hline
 $(u+1)(x+1)$  & $4$ & $2$  &$*$\\
\hline
 $(u+1)(x^4+x^3+x^2+x+1)$  & $1$ &  $5$ &$*$\\
\hline
 $x+1+u$ & $4$ &$1$  &\\
\hline
 $x^4+x^3+x^2+x+1+u$  & $4$ & $1$ & \\
\hline
\end{longtable}
\end{example}
\begin{example}
For length $n=6$.
$$x^6-1=(x+1)^3(x+2)^3 ~\text{over}~ \mathbb{Z}_3.$$
From Theorem \ref{Th1}, the non-zero free module or single generator reversible cyclic codes of length $6$ over $\mathbb{Z}_3 + u \mathbb{Z}_3 $ are given below:

\begin{longtable}[h]{|l|l|l|l|l|l|}
\hline
	  Non-zero Generator & Dimension $k$ &$d(C)$   & MDS \\
	 	Polynomial (s) of $C$ & of $C $ &  &\\
	 \hline
 $1$ or $(1+u)$  &$6$ & $1$& $*$\\
\hline
  $x+1$ & $5$ & $2$& $*$\\
\hline
  $(x+1)^2$  &$4$ & $2$& \\
\hline
  $(x+1)^3$  &$3$ & $2$& \\
\hline
  $(x+2)^2$  &$4$ & $2$&  \\
\hline
  $x+1, u$  &$6$ & $1$&  $*$\\
\hline
  $(x+1)^2, u$   &$6$ & $1$& $*$\\
\hline
  $(x+1)^3, u$  &$6$ & $1$&  $*$\\
\hline
 $(x+1)^2, u(x+1)$ &$5$ & $2$&  $*$\\
\hline
  $(x+1)^3, u(x+1)$  &$5$ & $2$&  $*$\\
\hline
 $(x+1)^3, u(x+1)^2$  &$4$ & $2$&  \\
\hline
  $(x+2)^2, u$  &$6$ & $2$&  \\
\hline
\end{longtable}

\end{example}
\begin{example} For length $n=4$.
$$x^4-1=(x+1)^4=f^4 ~\text{over}~ \mathbb{Z}_2.$$ From Theorem \ref{Th1}, the non-zero free module or single generator reversible cyclic codes of length $4$ over $\mathcal{R}$ are given below:
\newpage

\begin{longtable}[h]{|l|l|l|l|l|l|}
\hline
	  Non-zero Generator & Dimension $k$ &$d(C)$   & MDS \\
	 	Polynomial (s) of $C$ & of $C $ &  &\\
	 \hline
 $1$ or $1+u$ & $4$ & $1$ & $*$ \\
\hline
 $x+1$ & $3$ & $2$ & $*$ \\
\hline
 $x+1, u$ & $4$& $1$& $*$ \\
\hline
   $x+1+u$ & $3$& $2$& $*$ \\
\hline
  $x^2+1$ & $2$& $2$ & \\
\hline
 $x^2+1, u$ & $4$& $1$ & $*$ \\
\hline
	  $x^2+1, u(x+1)$ & $3$& $2$& $*$ \\
\hline

 $x^2+1+u$ & $2$& $2$ &\\
\hline
   $x^2+1+u, u(x+1)$ &$3$&  $2$& $*$ \\
\hline
  $x^2+1+u(x+1)$ & $2$& $2$ &\\
\hline
  $x^2+1+u, u(x+1)$ & $3$& $2$& $*$ \\
\hline
  $x^3+1$ & $1$& $2$ &\\
\hline
  $x^3+1, u$ & $4$& $1$& $*$ \\
\hline
 $x^3+1+u, u(x+1)$ & $3$& $2$& $*$ \\
\hline
  $x^3+1+u(x+1), u(x^2+1)$ & $2$& $2$&\\
\hline
 $x^3+1, u(x+1)$ & $3$& $2$& $*$ \\
\hline
 $x^3+1, u(x^2+1)$ & $2$& $2$&\\
\hline
\end{longtable}
\end{example}

\begin{example}
For length $n=6$.
$$x^6-1=(x+1)^2(x^2+x+1) ~\text{over}~ \mathbb{Z}_2.$$ From Theorem \ref{Th1}, the non-zero free module or single generator reversible cyclic codes of length $4$ over $\mathcal{R}$ are given below:

\begin{longtable}[h]{|l|l|l|l|l|l|}
\hline
	  Non-zero Generator & Dimension $k$ &$d(C)$   & MDS \\
	 	Polynomial (s) of $C$ & of $C $ &  &\\
	 \hline
 $1$ or $(1+u)$ &$6$ & $1$ & $*$\\
\hline
 $x+1$ & $5$& $2$ & $*$\\
\hline
  $x+1+u$  & $5$& $2$ & $*$\\
\hline
 $x^2+1$  & $4$& $2$ &\\
\hline
 $x^2+1, u$  & $6$& $1$ & $*$ \\
\hline
 $x^2+1, u(x+1)$  & $5$& $2$ & $*$\\
\hline
  $x^2+x+1$   & $4$& $2$ &\\
\hline
  $x^2+x+1, u$  & $6$& $1$ & $*$\\
\hline
  $x^3+1$ & $3$& $2$ & \\
\hline
  $x^3+1, u$  & $6$& $1$ & $*$\\
\hline
  $x^3+1, u(x+1)$  & $5$& $2$ & $*$\\
\hline
 $x^3+1+u, u(x+1)$  & $5$& $2$ & $*$\\
\hline
\end{longtable}
\end{example}
\begin{example}For length $n=7$.
$$x^7-1=(x + 1)(x^3 + x + 1)(x^3 + x^2 + 1) ~\text{over}~ \mathbb{Z}_2.$$

 The only self reciprocal factors are $(x-1)$ and $(x^6+x^5+x^4+x^3+x^2+x+1)$. By Theorem $\ref{6}$, reversible cyclic codes of length $7$ over $R$ are given below:
\begin{longtable}[h]{|l|l|l|l|l|}
\hline
	   Non-zero Generator & Dimension $k$ &$d(C)$   & MDS \\
	 	Polynomial (s) of $C$ & of $C $ &  &\\
	 \hline
 $1$ or $1+u$      & $7$ & $1$& $*$\\
 \hline
$(1+u)(x+1)$	 & $6$ &$2$ &$*$\\
\hline	
 $(x+1),u$    & $7$ & $1$& $*$\\
\hline
 $(x^6+x^5+x^4+x^3+x^2+x+1)$   & $1$ & $7$& $*$\\
\hline
 $(x^6+x^5+x^4+x^3+x^2+x+1)$, $u$  & $7$ & $1$& $*$\\
\hline
\end{longtable}
\end{example}
\begin{example} For length $n=10$.
$$x^{10}-1=(x+1)^5(x+4)^5 ~\text{over}~ \mathbb{Z}_5.$$ From Theorem \ref{Th1}, the non-zero free module or single generator reversible cyclic codes of length $10$ over $\mathbb{Z}_5+u \mathbb{Z}_5$ are given below:
\begin{longtable}[h]{|l|l|l|l|l|l|}
\hline
	   Non-zero Generator & Dimension $k$ &$d(C)$   & MDS \\
	 	Polynomial (s) of $C$ & of $C $ &  &\\
	 \hline
  $1$ or $1+u$ & $10$ &$1$& $*$\\
\hline
  $(x+1)^j, 1 \leq j \leq 5 $ & $9,8,7,6,5$ & $2 $& $*$, for $j=1$ \\
\hline
  $(x+4)^{2j}, ~ \text{where}~ j=1,2$ & $8
,6$ &$2$&\\
\hline
    $x+1, u$ & $10$ & $1$ & $*$\\
\hline
   $(x+1)^2, u$ &$10$ & $1$& $*$\\
\hline
  $(x+1)^2, u(x+1) $ &$9$ & $2$& $*$\\
\hline
  $(x+1)^3, u$  & $10$ & $1$& $*$\\
\hline
 $(x+1)^3, u(x+1)^j$ where $ 1 \leq j \leq 2$ & $9,8$ &$2$& $*$, for $j=1$\\
\hline
   $(x+1)^4, u$ &$10$ &  $1$& $*$\\
\hline
 $(x+1)^4, u(x+1)$ &$9$ & $2$& $*$\\
\hline
 $(x+4)^2, u $ &$10$ & $1$& $*$\\
\hline
   $(x+1)^2+ux$ &  $8$ & $2$& \\
\hline
 $(x+4)^2+ u x$ &  $8$ &$2$&\\
\hline
   $(x+1)^i(x+4)^{2},  $  where $ i=1,2$ &  $7,6$ &$3$&\\
\hline
  $(x+1)^3(x+4)^{2} $  &  $5$ &$4$&\\
\hline
   $(x+1)^4(x+4)^{2}  $  &  $4$ &$5$&\\
\hline
  $(x+1)^5(x+4)^{2}  $  &  $3$ &$6$&\\
\hline
   $(x+1)(x+4)^{4}  $  &  $5$ &$4$&\\
\hline
  $(x+1)^i(x+4)^{4}, $ ~\text{where}~ $ 2 \leq i \leq 4$ &  $4,3,2$ &$5$&\\
\hline
  $(x+1)^5(x+4)^{4}  $  &  $1$ &$10$& $*$\\
\hline
  $(x+1)^4+ u x^2, u$ &  $10$ &$1$& $*$\\
\hline
   $(x+1)^4+ u x^2, u(x+1)$ & $9$ &$2$& $*$\\
\hline
  $(x+1)^2+u x, u(x+1)^2$ & $8$ &$2$&\\
\hline
  $(x+1)^2(x+4)^2+u(x^3+x)$ & $6$ &$3$&\\
\hline
$(x+1)^2(x+4)^2+u(x^3+x^2+x)$ & $6$ &$3$&\\
\hline
$(x+1)(x+4)^2+u(x^2+x)$ & $7$ &$3$&\\
\hline
\end{longtable}
\end{example}
\hspace{-.7cm}\textbf{Remark:} In above examples $*$ represents the optimal (MDS) codes calculated by using magma software \cite{Bosma}.
\section{Conclusion}
In this article, we studied reversible cyclic codes of arbitrary length $n$ over the ring $ R = \mathbb{F}_q + u \mathbb{F}_q$, where $u^2=0$. We have provided a unique set of generators for these codes as ideals in the ring $R[x]/\langle x^n-1\rangle$. Moreover, in Section $5$, we have imposed some conditions under which dual of a reversible cyclic code is reversible. In Section $6$, we have given some examples in support of our results.

\end{document}